\documentclass[a4paper,UKenglish,cleveref,thm-restate,pdfa]{lipics-v2021}
\usepackage[draft]{preamble}

\title{SimdQuickHeap: The QuickHeap Reconsidered}

\author{Johannes Breitling}{Karlsruhe Institute of Technology, Germany}{johannes.breitling@student.kit.edu}{https://orcid.org/0009-0000-9706-0074}{}

\author{Ragnar Groot Koerkamp}{Karlsruhe Institute of Technology, Germany}{ragnar.grootkoerkamp@gmail.com}{https://orcid.org/0000-0002-2091-1237}{}
\author{Marvin Williams}{Karlsruhe Institute of Technology, Germany}{marvin.williams@kit.edu}{https://orcid.org/0000-0002-1990-0899}{}


\authorrunning{J. Breitling, R. Groot Koerkamp, and M. Williams}
\Copyright{Johannes Breitling, Ragnar Groot Koerkamp, and Marvin Williams}

\supplement{}

\supplementdetails[subcategory={Source Code}, cite=quickheap-repo, swhid={swh:1:dir:2cead1c11619177756143a9713ff21207d29fabd}]{Software}{https://github.com/RagnarGrootKoerkamp/quickheap}

\acknowledgements{
  We thank Peter Sanders for discussions related to this work and feedback on a
  draft of this paper.
}


\EventEditors{Philip Bille, Seth Pettie, and Sabine Storandt}
\EventNoEds{3}
\EventLongTitle{34th Annual European Symposium on Algorithms (ESA 2026)}
\EventShortTitle{ESA 2026}
\EventAcronym{ESA}
\EventYear{2026}
\EventDate{August 31--September 4, 2026}
\EventLocation{L'Aquila, Italy}
\EventLogo{}
\SeriesVolume{388}
\ArticleNo{16}

\begin{document}

\maketitle

\begin{abstract}%
\vspace{-2em}
\subparagraph*{Motivation.}
Priority queues are data structures that maintain a dynamic collection of elements and allow
inserting new elements and removing the smallest element.
The most widely known and used priority queue is likely the implicit binary heap,
even though it has frequent cache misses and is hard to optimize using e.g.
SIMD instructions.

\parag{Contributions.}
We introduce the SimdQuickHeap, a variant of the QuickHeap that was introduced by Navarro and Paredes in 2010.
As suggested by the name, the data structure bears some similarity to QuickSort.
We modify the data layout of the original QuickHeap
to have all \emph{pivots} adjacent in memory, with
elements between consecutive pivots stored in dedicated \emph{buckets}.
This allows efficient SIMD implementations for both partitioning of buckets
and scanning the list of pivots
to find the bucket to append newly inserted elements to.

The SimdQuickHeap has amortized expected complexity $O(\log n)$ per operation,
which improves to $O(\frac 1W\log n)$ in non-degenerate cases, where $W$ is the number of words in a SIMD register.
In this case, the I/O-complexity is amortized $O(\frac 1B)$ per \Op{push} and $O(\frac 1B \log_2 \frac nM)$ per \Op{pop}.

\parag{Results.}
In synthetic benchmarks, the SimdQuickHeap is $1.2\times$ to $1.7\times$ as fast
as the monotone radix heap, the next-best competitor, and $1.4\times$ to $2.8\times$ as fast as the superscalar sample queue, the fastest comparison-based priority queue.
The SimdQuickHeap needs around $1.5\log_2 n$
comparisons and $\log_2 n$ nanoseconds per pair of push and pop operations.
On graph benchmarks with Dijkstra's shortest path algorithm and Jarn{\'i}k--Prim's
minimum spanning tree algorithm, the SimdQuickHeap is consistently the fastest.
\end{abstract}


\section{Introduction}\label{intro}
The Priority Queue (PQ) is a fundamental data structure that manages a collection of elements with associated priorities, allowing the insertion of new elements and the extraction of the element with the highest priority.
PQs are central to a wide range of algorithms that rely on dynamically reordering and prioritizing tasks, most famously Dijkstra's shortest-path algorithm~\cite{dijkstra59}, but also job scheduling, discrete event simulation, and branch-and-bound searches.
As a result, PQs have been studied extensively, resulting in a large variety of priority queue designs (see \cref{sec:related}).

Despite being among the oldest designs, (implicit) binary heaps are probably still the most widely used priority queues in practice\footnote{They are available in the standard libraries of many programming languages, including C\texttt{++}, Java, C\#, Go, Rust, and Python.}.
There are compelling reasons for their prevalence: they are conceptually simple, easy to implement, and have $\Theta(\log_2 n)$ worst-case running time with small constants for both insertions and deletions.
While many theoretically superior designs have been proposed, most of them are prohibitively complex for practical use, and surprisingly few designs that are faster in practice have emerged.

One such design is the \emph{QuickHeap}~\cite{quickheap}, introduced in 2010 by Navarro and Paredes.
It uses a recursive partitioning scheme reminiscent of quicksort~\cite{quicksort}, hence the name%
\footnote{Although not a heap in the tree-like sense, the layers between the pivots (\cref{fig:heaps}) do form a linear tree with unordered sets of elements as nodes.}.
However, it has received little attention in the literature.
\Cref{fig:heaps} shows a schematic overview of the QuickHeap among other priority queue designs.
We show that, due to its partitioning scheme, the QuickHeap lends itself to the SIMD (\emph{single instruction, multiple data}) paradigm much more naturally than binary heaps and most other heaps based on merging trees or streams of elements.
SIMD instructions are steadily becoming faster, more efficient, more powerful, and are capable of operating on wider registers.
In light of these advancements, we argue that the QuickHeap should be reconsidered.

\begin{figure}[h]
	\centering
	\includegraphics[width=.9\linewidth]{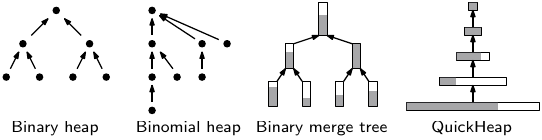}
	\caption{\label{fig:heaps}%
		Overview of some heaps.
		Black dots indicate single elements, and arrows point from larger to smaller elements.
		Vertical blocks indicate a sorted list of elements, while horizontal blocks indicate an unsorted list of elements.
		The QuickHeap stores pivots separating the buckets explicitly.
	}
\end{figure}

\parag{Contributions.}
We present \emph{SimdQuickHeap}, an adaptation of the QuickHeap that achieves significantly improved performance through optimized data layout and SIMD instructions.
To our knowledge, it is the first practical SIMD-optimized priority queue implementation.

The \Op{push} and \Op{pop} operations have an amortized worst-case run time bound of $\BigO(\log n)$ and $\BigO(1)$, respectively.
In most practical cases, the SIMD implementation speeds up both operations by a factor of the SIMD register width \(W\).
Like the QuickHeap, the SimdQuickHeap is cache-oblivious with an amortized
I/O-complexity of amortized $\BigO(\tfrac 1B \log\tfrac nM)$ per
operation for a main memory of $M$ words and block size $B$, assuming $M =\Omega(\log n)$.

On synthetic benchmarks, the SimdQuickHeap is $1.2\times$ to $1.7\times$ as fast as
a state-of-the-art implementation of the radix heap~\cite{radix-heap-cpp} and
$1.4\times$ to $2.8\times$ as fast as the superscalar sample queue~\cite{superscalar-queue}.
On real-world instances of Dijkstra's shortest path algorithm~\cite{dijkstra59} and Jarn{\'i}k--Prim's minimum-spanning-tree algorithm~\cite{boruvka26,jarnik30,prim57}, it is consistently the fastest method.
Notably, while the run time of most tree-based and pointer-based priority queues grows relative to the $\Omega(\log_2 n)$ lower bound as $n$ increases, SimdQuickHeap's run time improves, falling below $\log_2 n$ nanoseconds per \Op{push}-\Op{pop} pair.

\parag{Notation and definitions.}
Given a universe of elements \(\mathcal{U}\) and a strict weak ordering \(<\) on \(\mathcal{U}\cup \{-\infty,\infty\}\), the priority queue (PQ) data structure is a dynamic collection of elements \(\mathcal{Q}\) over \(\mathcal{U}\) (with duplicates permitted) that supports the following two operations:
\begin{description}
	\item[push(e)] Insert an element \(e\in \mathcal{U}\) into \(\mathcal{Q}\).
	\item[pop()] Remove and return a minimal element \(e\in\mathcal{Q}\), i.e., no \(e'\in \mathcal{Q}\) with \(e'<e\) exists.
\end{description}

For convenience, we assume \Op{top} (\Op{pop} without removing) and \Op{size} (returns \(\lvert\mathcal{Q}\rvert\)) to also be available.
We do not consider a \Op{decreaseKey} operation that replaces
arbitrary elements in the PQ by smaller ones, since it usually requires stable references to elements, adding implementation complexity and run-time overhead.
While our definition admits distinct, equivalent elements, we refer to all equivalent elements as \emph{equal} elements for simplicity.

\section{The Original QuickHeap}\label{sec:quickheap}
\enlargethispage{2em}
In 2006, Paredes and Navarro presented \emph{incremental quick select}~\cite{optimal-incremental-sorting}, an algorithm that takes a list of $n$ elements and then repeatedly returns the next-smallest element, such that returning the $k$ smallest elements in an online setting takes $\BigO(n+k\log k)$ time.
The authors note that their construction can be generalized to a proper priority queue.
Consequently, in follow-up work from 2010, they present the QuickHeap~\cite{quickheap}.
The QuickHeap organizes the elements in disjoint \emph{buckets} \(B_1,\ldots,B_\ell\) that are separated by \emph{pivots} \(p_1>\ldots>p_{\ell-1}\) with
\[
	\infty \succ B_1 \succeq p_1 \succ B_2 \succeq p_2 \succ \dots \succeq p_{\ell-1} \succ B_\ell \succ -\infty,
\]
where $\succeq$ ($\succ$) indicates that the $\geq$ ($>$) inequality holds for all elements in the sets.
For the sake of simplicity, in the following we assume that the PQ contains no two equal elements.
See \cref{scalar-design} for details on how equal elements are handled in our design.
The \Op{pop} operation picks a new pivot \(p_\ell\) (e.g., uniformly at random) from the \emph{bottom} bucket \(B_\ell\) and moves all elements in \(B_\ell\) smaller than \(p_\ell\) into a new bucket \(B_{\ell+1}\).
If the new bucket is empty, it is discarded and a new pivot is selected.
This process is repeated recursively in \(B_{\ell+1}\) until the bottom bucket contains only one (the smallest) element.
Finally, this element is returned and its bucket is removed.
The \Op{push} operation simply inserts the element into the correct bucket according to the pivots.

The original implementation~\cite{quickheap} stores all buckets and pivots in an interleaved fashion in a flat array, with an additional array to track the positions of the pivots.
Splitting a bucket is implemented analogously to partitioning in quicksort.
To make room for a new element in a bucket \(B_i\), all buckets with index smaller than \(i\) are shifted one slot to the right by moving two elements per bucket.
The data structure is stored in a circular buffer, so that removing the smallest element only requires incrementing a pointer.
\Cref{fig:quickheap-old2} shows the data layout and illustrates the \Op{push} and \Op{pop} operations.

\begin{figure}[b]
	\centering
	\includegraphics[width=.9\linewidth]{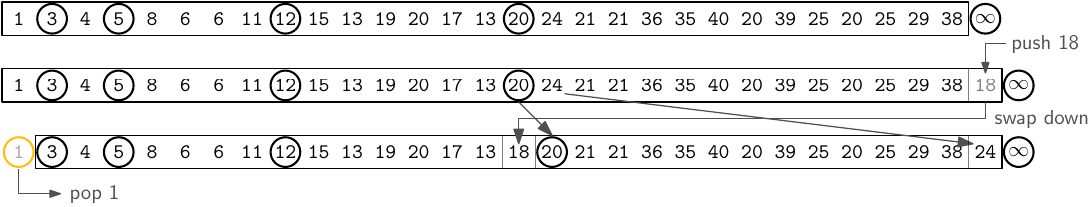}
	\caption{\label{fig:quickheap-old2}%
		Overview of the QuickHeap data structure~\cite{quickheap}.
		The data is stored in a single buffer.
		Pivots are circled and partition the data into smaller (left) and larger (right) elements. Their positions are stored in a separate list (not shown).
		An element is pushed by appending it at the end and then sifting it down until it is less than the corresponding pivot.
		An element is popped by repeatedly partitioning the bottom layer and then removing the smallest element.
	}
\end{figure}

The \Op{push} and \Op{pop} operations of the QuickHeap are both in \(\BigO(\log n)\) \emph{in expectation over uniformly distributed input}~\cite{quickheap}:
the original QuickHeap can degenerate under specific input distributions.
For example, when the inserted elements are decreasing, the number of buckets/pivots grows linearly over time, causing a run time of \(\BigO(n)\) for insertions.
Follow-up work introduces \emph{Randomized QuickHeaps}~\cite{stronger-quickheaps}, which addresses this issue by probabilistic repartitioning:
each time a bucket of size \(s\) is shifted during an insertion, with probability \(1/s\), the Randomized QuickHeap unites the bucket with all newer buckets by dropping the corresponding pivots.
While the Randomized QuickHeap achieves \(\BigO(\log n)\) operations in expectation for \emph{any} input, it adds significant overhead for benign cases.
Very recently, Brodal et al.~\cite{partition-based-simple-heaps} introduced more sophisticated rebalancing schemes (see \cref{sec:related}).
However, these are not directly practical since they rely on linked lists to split and join buckets efficiently.

\section{Related work}\label{sec:related}
There is a vast body of literature on priority queues that dates back to the 1950s.
Here, we highlight work on practical priority queues and focus on RAM-model complexity, I/O-complexity, the number of comparisons, and the running time.
For a broader overview, we refer to the survey by Brodal~\cite{priority-queue-survey} from 2013.
A survey on the performance of classic priority queues is given by Larkin et al.~\cite{back-to-basics-priority-queues}.
The QuickHeap does not appear in either of these surveys, highlighting the little attention it has received.

\parag{Complexity lower bounds.}
There is a strong connection between sorting and priority queues~\cite{equivalence-priority-queues-sorting}.
Due to the lower bound for comparison based sorting, \(n\) insertions followed by \(n\) deletions must take at least \(n\log_2 n\) comparisons for a comparison-based priority queue.
A \Op{pop} followed by a \Op{push} requires at least \(\log_2 n\) comparisons in the worst case amortized over many operations~\cite{complexity-of-priority-queue-maintenance}.

\parag{Binary and $d$-ary heap.}
A popular priority queue implementation is the binary heap using an implicit array representation~\cite{heapsort}.
The binary heap organizes the elements as a binary tree, maintaining the \emph{heap invariant} that no element can be larger than its parent in the tree (unless it is the root).
Thus, no element can be smaller than the root element.
A binary heap can be implemented efficiently using the \emph{Eytzinger} or \emph{heap layout}, where the elements are stored consecutively in a flat array level by level starting from the root.
Consequently, the children of node \(i\) are stored at positions \(2i\) and \(2i+1\) in the array (\(1\)-based indexing).
Both the $\push$ and $\pop$ operations take \(\mathcal{O}(\log_2(n))\) time.

The \(d\)-ary heap generalizes the binary heap so that each node has up to \(d\) children.
The main benefits are better cache locality and fewer levels to traverse.
This comes at the cost of needing to find the minimum of \(d\) elements on each level when deleting an element.

\parag{Amortized heaps.}
A number of heaps reduce the $\BigO(\log n)$ cost of binary heap insertions.
Binomial heaps (\cref{fig:heaps})~\cite{binomial-heap,binomial-heap-analysis} have amortized $\BigO(1)$ insertions. Fibonacci heaps~\cite{fibonacci-heap} and pairing heaps~\cite{pairing-heap} reduce this to truly constant time insertions, but have only amortized constant time deletions.
The Brodal heap~\cite{brodal-queue} (1996) has $\BigO(1)$ insertions and non-amortized $\BigO(\log n)$ deletions, but has a large hidden constant.
Other non-amortized structures are the cascade heap~\cite{cascade-heap}, which has $\BigO(1)$ insertions and $\BigO(\log^* n + \log k)$ for the $k$'th delete, and the logarithmic funnel~\cite{log-funnel-heap}, which has $\BigO(\log^* n)$ insertions and $\BigO(\log n)$ deletions.

Unfortunately, all these structures are pointer-based and inefficient in practice.
The weak heap and weak queue are more practical variants of these ideas that can be implemented using a flat array~\cite{weak-heap-sort,benchmarking-weak-heaps,improved-weak-heap,engineering-weak-heaps}.

\parag{decreaseKey.}
Some of the priority queues mentioned above also support \Op{decreaseKey}, which lowers the key associated with a given element.
Typically (but not always~\cite{buffer-heap}), data structures supporting this require \emph{pointer stability} of the elements and thus an additional level of indirection, hurting performance.
We do not consider these further, but note that in the pointer-based model there is a
lower bound of $\Omega(\log \log n)$
for \Op{decreaseKey}~\cite{decrease-key-lower-bound} that is attained by slim and smooth heaps~\cite{slim-smooth-heap-analysis,self-adjusting-heaps}.

\parag{I/O-model.}
Given that optimal priority queues for the RAM-model were already found early on, further research went into obtaining a good I/O-complexity~\cite{io-complexity-sorting}.
In this setting, we are given an $M$-word internal memory and must minimize the number of transfers of $B$-word blocks to/from external memory.
The lower bound for priority queues matches that of sorting, $\Omega(\tfrac 1B \log_{M/B} \tfrac nB)$ amortized transfers per operation~\cite{io-complexity-sorting}, and a data structure matching this is given by Jiang and Larsen~\cite{external-memory-priority-queue-decreasekey}.
Practical implementations of I/O-optimal priority queues are the array
heap~\cite{array-heap} and sequence-heap~\cite{sequence-heap}, which improves
the constant factor.
The QuickHeap has slightly higher I/O cost of \(\BigO((1/B) \log_2 (n/M))\) for the \Op{push} and \Op{pop} operations~\cite{quickheap}.
With additional support for \Op{decreaseKey}, the lower bound increases to $\Omega(\tfrac 1B \log B / \log \log n)$~\cite{decreasekey-expensive}, and $\BigO(\tfrac 1B \log \tfrac nB /\log\log n)$ is achieved by~\cite{external-memory-priority-queue-decreasekey}.
A \emph{cache-oblivious}~\cite{cache-oblivious-algorithms} priority queue with
optimal amortized complexity that does not rely on the value of $M$ and $B$ is given by~\cite{cache-oblivious-priority-queue,cache-oblivious-priority-queue-extended}, with improvement to $\BigO(1/B)$ insertions in~\cite{external-memory-priority-queue-optimal-insertions}.
The buffer heap~\cite{buffer-heap-spaa,buffer-heap} supports \Op{decreaseKey} in $\BigO(\tfrac 1B \log_2 \tfrac nM)$ transfers and was discovered in parallel with the bucket heap~\cite{bucket-heap}.
The funnel heap~\cite{funnel-heap} is another optimal cache-oblivious queue that relies on \emph{binary mergers} (\cref{fig:heaps}) with strategically chosen buffer sizes.


\parag{Other priority queues.}
There are also data structures that exploit integer keys.
These data structures are typically limited to workloads where the deleted elements increase monotonically.
We refer to the survey by Costa et al.~\cite{monotone-priority-queue-review} for an overview.
The radix heap~\cite{radix-heap} inserts keys into a bucket based on the largest
bit in which they differ from the current minimum, and has amortized complexity
$\BigO(\log C)$ per operation when inserted keys are at most $C$ larger than the
current minimum. Like the QuickHeap, bucket sizes grow roughly exponentially.
Further examples are bucket queues~\cite{bucket-queue,two-level-bucket-queue} (similar to bucket sort) and
$\BigO(\log w)$ dynamic predecessor structures for $w$-bit keys like Van Emde Boas trees~\cite{van-emde-boas-trees}, Y-fast tries~\cite{xy-fast-trie}, and fusion trees~\cite{fusion-trees}.

\parag{Practical/engineered priority queues.}
Given the plethora of theoretical results, there are surprisingly few engineered implementations.
Sequence heaps (2000)~\cite{sequence-heap} provide an optimized I/O-efficient implementation based on merging sorted lists, with the drawback that elements are eagerly sorted, even when they are never removed.
The QuickHeap (2010)~\cite{quickheap} is indeed much faster when only few elements are removed, and competitive on synthetic benchmarks, but in theory is less I/O-efficient than the sequence heap.
The superscalar sample queue (S${}^3$Q, 2021)~\cite{superscalar-queue} is a
much more recent engineered data structure
based on super scalar sample sort~\cite{super-scalar-sample-sort}, that uses $k$-way partitioning instead of $k$-way merging to reach the optimal I/O-complexity.
In particular, partitioning tends to be easier to optimize than merging, even though $k$-way splitting is harder to do evenly than data-independent $k$-way merging.

One further engineered data structure is the B-heap~\cite{b-heap}, which implements an implicit binary heap using a recursive layout like the Van Emde Boas tree~\cite{van-emde-boas-trees}.

\enlargethispage{-2em}

\parag{Partition-based heaps.}
Most methods above are based on variously shaped trees with the heap property.
Some methods, like the QuickHeap, instead use a
path-shaped (linear) heap of buckets, where pivots (implicitly) partition the
data into buckets.
The priority queue of~\cite{cache-oblivious-priority-queue,cache-oblivious-priority-queue-extended} uses exponentially growing layers where each layer contains an increasing number of sorted buffers of increasing size.
The buffer heap~\cite{buffer-heap-spaa,buffer-heap} is conceptually simpler: each layer contains an unsorted list of elements.
The superscalar sample queue combines these ideas, and each level contains $k$ unsorted buffers.
Brodal et al. \cite{partition-based-simple-heaps} recently revisited this
concept and provide \emph{rebalancing strategies} to ensure the number of layers is
guaranteed to be logarithmic.
In one strategy, two consecutive buckets $B_{i-1}$ and $B_i$ are merged after a $\pop$ if they contain fewer elements in total than all previous buckets combined, i.e., $|B_i| + |B_{i-1}| < |B_1| + |B_2| + \dots + |B_{i-2}|$.
Under this rule, the number of layers is bounded by $1+2\log_2 n$.
Another strategy is to maintain the invariant $|B_i| < 3\cdot 2^i$, resulting in at most $1+\log_2 n$ layers.
Both strategies use a linear time pivot selection algorithm to partition the buckets into the right sizes.
Elements are stored in a single linked list so that merging partitions takes constant time.




\section{The new SimdQuickHeap}\label{simd-quickheap}
We first describe the basic scalar design of the \emph{SimdQuickHeap} based on the original QuickHeap (\cref{sec:quickheap}) and analyze its run time.
We then discuss improvements using SIMD.

\subsection{Scalar Design}\label{scalar-design}
\begin{figure}
	\centering
	\includegraphics[width=.9\linewidth]{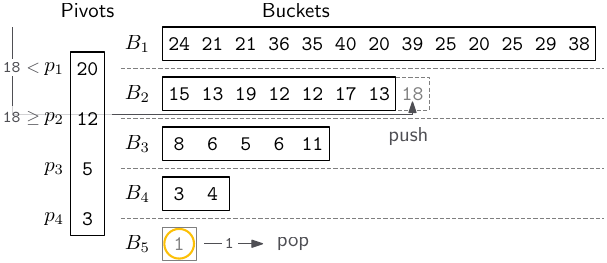}
	\caption{\label{quickheap-new}%
		A schematic overview of the SimdQuickHeap:
		each bucket is a separate array, and
		the pivots are stored in their own array in sorted order.
	}
\end{figure}
Our variant of the QuickHeap keeps an explicit array of all pivots in decreasing order, as shown on the left in \cref{quickheap-new}.
The buckets are stored in separate buffers, rather than as a single allocation.
The \Op{push} operation classifies each element using a binary search on the pivot
array and then appends it to the respective bucket.
To partition the bottom bucket for the \Op{pop} operation,
a \emph{pivot} is selected and
all elements smaller than the pivot are moved to a newly allocated bucket, while
compactifying the remaining elements in-place, as shown in \cref{alg:pushpop}.
In case we selected the smallest element in the bucket as a pivot, no element moves to the new bucket, so we pick a new pivot and start over.
However, if all elements in the bucket are equal, this does not terminate.
For this reason, we treat elements equal to the pivot specially.
\begin{algorithm}[t]
	\centering
	\caption{Pseudo code for the \texttt{SimdQuickHeap} data structure.
		The \Op{pop} operation assumes that it is not empty.\label{alg:pushpop}}
	\begin{algorithmic}
		\State \textbf{structure} $\textproc{SimdQuickHeap}\langle T \rangle$
		\State \hspace{\algorithmicindent} $\mathit{pivots}: \texttt{Vec}\langle T \rangle$, a vector of decreasing pivots
		\State \hspace{\algorithmicindent}
		$\mathit{buckets}: \texttt{Vec}\langle \texttt{Vec}\langle T \rangle
			\rangle$, a vector of bucket vectors
	\end{algorithmic}
	\begin{algorithmic}[1]
		\Procedure{push}{$x$}
		\State $i \gets \textproc{classify}(x, \mathit{pivots})$\Comment{Binary
			search the first pivot $\leq x$.}
		\State $\mathit{buckets}[i].\textproc{pushBack}(x)$
		\EndProcedure
	\end{algorithmic}
	\begin{algorithmic}[1]
		\Procedure{pop}{}
		\While{$\mathit{buckets}.\textproc{last}().\textproc{len}() > 1$}
		\Comment{Partition the bottom layer repeatedly.}
		\State $(p, B') \gets \textproc{Partition}(\mathit{buckets}.\textproc{last}())$
		\If{$B'\neq\emptyset$}
		\State $\mathit{buckets}.\textproc{pushBack}(B')$
		\State $\mathit{pivots}.\textproc{pushBack}(p)$
		\EndIf
		\EndWhile
		\State $x \gets \mathit{buckets}.\textproc{last}().\textproc{popBack}()$
		\State $\mathit{buckets}.\textproc{popBack}()$
		\State $\mathit{pivots}.\textproc{popBack}()$
		\State \Return $x$
		\EndProcedure
	\end{algorithmic}
	\begin{algorithmic}[1]
		\Procedure{partition}{$B$}
		\State $(i,p) \gets \textproc{SelectPivot}(B)$ \Comment{$p=B[i]$}
		\State $k \gets 0$
		\State $B' \gets \texttt{Vec}\langle T \rangle()$
		\For{$j\in [1,B.\textproc{len}()]$}
		\If{$B[j] < p$ \textbf{or} $(j>i \textbf{ and } B[j] \leq p)$}
		\State $B'.\textproc{pushBack}(B[j])$
		\Else
		\State $k \gets k+1$
		\State $B[k] \gets B[j]$
		\EndIf
		\EndFor
		\State $B.\textproc{Resize}(k)$
		\State \Return $(p, B')$
		\EndProcedure
	\end{algorithmic}
\end{algorithm}

\parag{Equal elements.} If all elements in a bucket are equal, we need to ensure that partitioning terminates and does not degenerate to splitting off one element at a time.
To achieve this, elements in the bucket that are on the left of the pivot are
moved down when they are \emph{strictly} smaller than the pivot, while elements
on the right of the pivot are moved down when they are smaller \emph{or equal}.
This way, in expectation half the (non-pivot) elements are moved to the new bucket.
Since the pivot itself remains in its bucket, we maintain the invariant that no bucket is ever empty, unless the data structure is completely empty, in which case there is only one empty bucket.
Alternatively, one could use dedicated buckets for elements equal to the pivot, but we did not pursue this strategy further.
Note that our strategy allows for elements equal to the pivot $p_i$ to be in buckets $B_i$ and $B_{i+1}$.

\enlargethispage{1em}
\parag{Pivot selection.} As in quicksort, the performance of the SimdQuickHeap depends on the pivot quality.
The median would be the ideal pivot, as it splits the bucket into two equal parts.
Experimentally, the median of three random samples offers good performance across all tested workloads.
In particular, it is typically a few percent faster than choosing a random pivot or the median of five.
Stronger strategies, such as the true median or median-of-medians~\cite{blumTimeBoundsSelection1973}, did not justify their additional overhead.

\parag{Rebalancing.}
The SimdQuickHeap does not suffer from degenerate inputs as much as the original QuickHeap (see \cref{sec:quickheap}):
While the original QuickHeap ``bubbles down'' elements to insert them, the SimdQuickHeap uses binary search to classify elements, which takes $\BigO(\log n)$ time no matter the number of pivots.
Thus, we do not employ any rebalancing strategies in our implementation.
However, we will consider rebalancing in future work, since it still has many potential benefits:
When there are $\BigO(\log n)$ buckets, the binary search has complexity
$\BigO(\log \log n)$ \cite{partition-based-simple-heaps}.
Alternatively, this allows for a simple linear scan in $\BigO(\log n)$.

\parag{I/O-complexity.}
When the number of layers is indeed $\BigO(\log n)$, the SimdQuickHeap has the
same amortized I/O-complexity as the QuickHeap: $\BigO(\tfrac 1B \log_2\tfrac
	nM)$ per operation, assuming that the first size-$B$ block of each bucket fits in the cache, i.e., $M =
	\Omega(B \log n)$.
Specifically, classifying a pushed element then does not require any memory transfers
since the pivots are in memory, and appending it to a block incurs amortized
$\BigO(\tfrac 1B)$ transfers by buffering the last block of elements appended to
each bucket.
The memory access patterns for the partitioning step are similar to the original
QuickHeap, and the original analysis applies.

\subsection{Amortized analysis}
For simplicity, we also assume that the true median is used as a pivot (as this
can be found in linear time), and that no two elements are equal.
Using the median of three random elements yields the same bounds for \Op{push} and, in expectation over the random pivot selection, also for \Op{pop}.
This can be shown with a similar, more technical, proof that uses the fact that median of three random elements gives a constant-factor split in expectation.
\begin{theorem}
	For the \textsf{SimdQuickHeap} containing \(n\) elements, the \Op{push} operation takes \(\BigO(\log n)\) amortized and worst-case time, the \Op{pop} operation takes \(\BigO(1)\) amortized time.
\end{theorem}
\begin{proof}
	The \Op{push} operation can classify the element to insert in \(\BigO(\log n)\) time via a binary search on the pivots, since there are at most \(\ell \leq n-1\) pivots.
	Appending the element to the correct bucket takes constant time, resulting in a total worst-case run time of $\BigO(\log n)$.

	For the amortized run time analysis, we define the \emph{potential function}
	\[
		\Phi \coloneq c\cdot \sum_{i=1}^\ell (\lvert B_i\rvert+1)\ln (\lvert B_i \rvert + 1)
	\]
	for some sufficiently large constant \(c\).
	Let \(k_i=\lvert B_i \rvert + 1\) be the \emph{weight} of the bucket \(B_i\).
	Then, inserting an element into bucket \(B_i\) increases the potential by
	\begin{align*}
		\Delta_{\textrm{push},i}\Phi & = c(k_i+1)\ln(k_i+1) - ck_i\ln(k_i)                                              \\
		                             & = ck_i\ln\left(1+k_i^{-1}\right)+c\ln(k_i+1)                                     \\
		                             & \leq c(1 +\ln(k_i+1))                        & \textrm{using \(\ln(1+x)\leq x\)} \\
		                             & \leq 2c\ln(n+2) \in \BigO(\log n).
	\end{align*}
	Thus, the \Op{push} operation takes \(\BigO(\log n)\) amortized time.

	The \Op{pop} operation first finds the median of the bottom bucket of size
	$k_\ell$ as pivot and partitions it, resulting in two buckets of equal size $k_\ell/2$.
	This takes \(\BigO(k_\ell)\) time and changes the potential (ignoring rounding errors) by
	\begin{align*}
		\Delta_{\textrm{part}}\Phi & = 2\left(c\frac{k_\ell}{2}\ln\left(\frac{k_\ell}{2}\right)\right) - ck_\ell\ln(k_\ell) = -ck_\ell\ln(2).
	\end{align*}
	Repeatedly partitioning until the bottom bucket contains only one element changes the potential (again, ignoring rounding errors) by \(\Delta_{\textrm{pop}}\Phi=-c\ln(2)(k_\ell+\tfrac{k_\ell}{2}+\tfrac{k_\ell}{4}+\ldots+2)=-\Omega(ck_\ell)\) and takes \(\BigO(k_\ell)\) time.
	Thus, with \(c\) sufficiently large, the potential can ``pay'' for the \(\BigO(k_\ell)\) actual cost of the \Op{pop} operation, giving \(\BigO(1)\) amortized time.
\end{proof}

\subsection{Practical Optimizations}
Here, we describe practical optimizations that take advantage of the data layout of the \textsf{SimdQuickHeap}.
Let $W$ be the number of elements that fit into one SIMD register.
When the number of layers is indeed $\BigO(\log n)$, the optimizations result in an $\BigO(\frac 1W \log n)$ algorithm.

\parag{Classification.}
If the number of pivots is small, a linear scan can be faster than a binary search, thanks to better cache efficiency and branch predictions.
We therefore employ a SIMD-optimized linear scan unless the number of pivots surpasses some threshold (64 by default), above which we fall back to binary search.
In practice, the linear scan is used for all feasible inputs where the number of buckets is in $\BigO(\log n)$ with this threshold.
We compare $W$ consecutive pivots at once with SIMD, resulting in a run time of
\(\BigO(\tfrac 1W \log n)\) when the number of pivots is in \(\BigO(\log n)\).

\parag{Partitioning.}
Multiple partitioning schemes using SIMD have been proposed \cite{bramasNovelHybridQuicksort2017,wassenbergVectorizedPerformanceportableQuicksort2022a}.
We use a straightforward mechanism:
We load $W$ values from the bucket and compare them against the pivot at once in $\BigO(1)$ time.
Using AVX2 \texttt{permutevar}\footnote{See
	\url{https://lemire.me/blog/2017/04/10/removing-duplicates-from-lists-quickly/}.}
($\sim 3$ cycles) or AVX-512 \texttt{compress} ($\sim 6$ cycles) instructions,
we can then efficiently move the values smaller (or larger) than the pivot
to the front of the register.
Finally, we append these values to the back of the bucket array.
This way, partitioning a bucket of size $m$ takes $\BigO(m/W)$ time.

In practice, the constant overheads of this partitioning scheme become too large
when the bucket is small (of length $\BigO(W)$).
Thus, we stop the iterative partitioning when the bucket size falls below length
$16$.
To be able to delete the minimum in the last bucket efficiently, we keep it sorted in descending order as long as it is smaller than the threshold.
Consequently, when a new element is inserted into the last bucket and the bucket size is below the threshold, the element is inserted at the appropriate position using a linear scan.

\parag{Further implementation details.}
We pre-allocate the pivot array with 128 entries, which is enough for most non-degenerate workloads.
We do not pre-allocate the bucket arrays but ensure that the capacity is always a multiple of the SIMD width to avoid special cases when reading the last elements.
Logically removed buckets are not freed but re-used when new buckets are needed.

\section{Experiments}\label{results}
The implementation of the SimdQuickHeap and the evaluations are written in Rust
and are available at \href{https://github.com/ragnargrootkoerkamp/quickheap}{github:ragnargrootkoerkamp/quickheap}.
We use an AMD Zen 4 EPYC 9684X (92 cores, 32 KiB L1 data cache, 1 MiB L2 cache, and 32+64 MiB L3 cache for each package of 8 cores) with a 12-channel DDR5 RAM running Rocky Linux 9.4 for the experiments.
All experiments are run using a single thread and repeated three times after one warm-up iteration.
We report the median of these runs but note that the measurements generally had very little variance.
We compile the rust code in release mode with \texttt{target-cpu=native} and enable full link-time optimization
with a single codegen unit to minimize the impact of cross-language function calls.
The \texttt{C++} code is compiled with \texttt{Clang} using the flags \texttt{-std=c++17 -O3 -DNDEBUG -march=native -flto}.
We further reimplemented some of the benchmarks in \texttt{C++} for the \texttt{C++} implementations, but measured no difference in performance.


\parag{Compared implementations.}
For \textsf{SimdQuickHeap}, we test both 256-bit AVX2 and 512-bit AVX-512 implementations.
We compare against a number of external implementations:
\begin{itemize}
  \item \textsf{BinaryHeap}: the Rust standard library \texttt{std::collections::BinaryHeap}.
  \item \textsf{Orx-$d$-aryHeap}: implementation from the \texttt{orx\_priority\_queue} Rust crate~\cite{orx-priority-queue}, the fastest $d$-ary heap rust implementation we found.
  \item \textsf{WeakHeap}: from the \texttt{weakheap} Rust crate~\cite{weak-heap-crate}.
  \item \textsf{RadixHeap}: the fastest radix heap we found, in \texttt{C++}~\cite{radix-heap-cpp}.
  \item \textsf{SequenceHeap}: the original \texttt{C++} implementation~\cite{sequence-heap,sequence-heap-code}.
  \item \textsf{S3qHeap}: the original super-scalar sample queue \texttt{C++} implementation~\cite{superscalar-queue,superscalar-queue-code}.
  \item \textsf{OriginalQuickHeap}: the Randomized QuickHeap implementation of \cite{stronger-quickheaps}.
\end{itemize}
Further implementations that are less efficient than the ones above are:
\hyphenation{Boost-Skew-Heap}
\begin{itemize}
  \item \textsf{$d$-aryHeap}: implementation from the \texttt{dary\_heap} Rust crate~\cite{dary-heap}.
  \item \textsf{Boost-$d$-ary heap}, \textsf{BoostFibonacciHeap}, \textsf{BoostPairingHeap}, \textsf{BoostBinomialHeap}, \textsf{BoostSkewHeap}: various heap implementations from the \texttt{C++} Boost library \cite{boost}.
  \item \textsf{RadixHeapMap}: from the \texttt{radix\_heap} Rust crate~\cite{radix-heap-crate}.
  \item \textsf{PairingHeap}: from the \texttt{pheap} Rust crate~\cite{pheap-crate}.
  \item \textsf{FibonacciHeap}: from the \texttt{fibonacci\_heap} Rust crate~\cite{fibonacci-heap-crate}.
  \item \textsf{ReimplementedQuickHeap}: our reimplementation of the QuickHeap.
\end{itemize}

We further implemented a \textsf{ScalarQuickHeap} variant of the \textsf{SimdQuickHeap} that does not use SIMD instructions and can be instrumented to count internal metrics such as the number of comparisons.

\begin{figure}[tp]
	\centering
	\includegraphics[width=\textwidth,keepaspectratio]{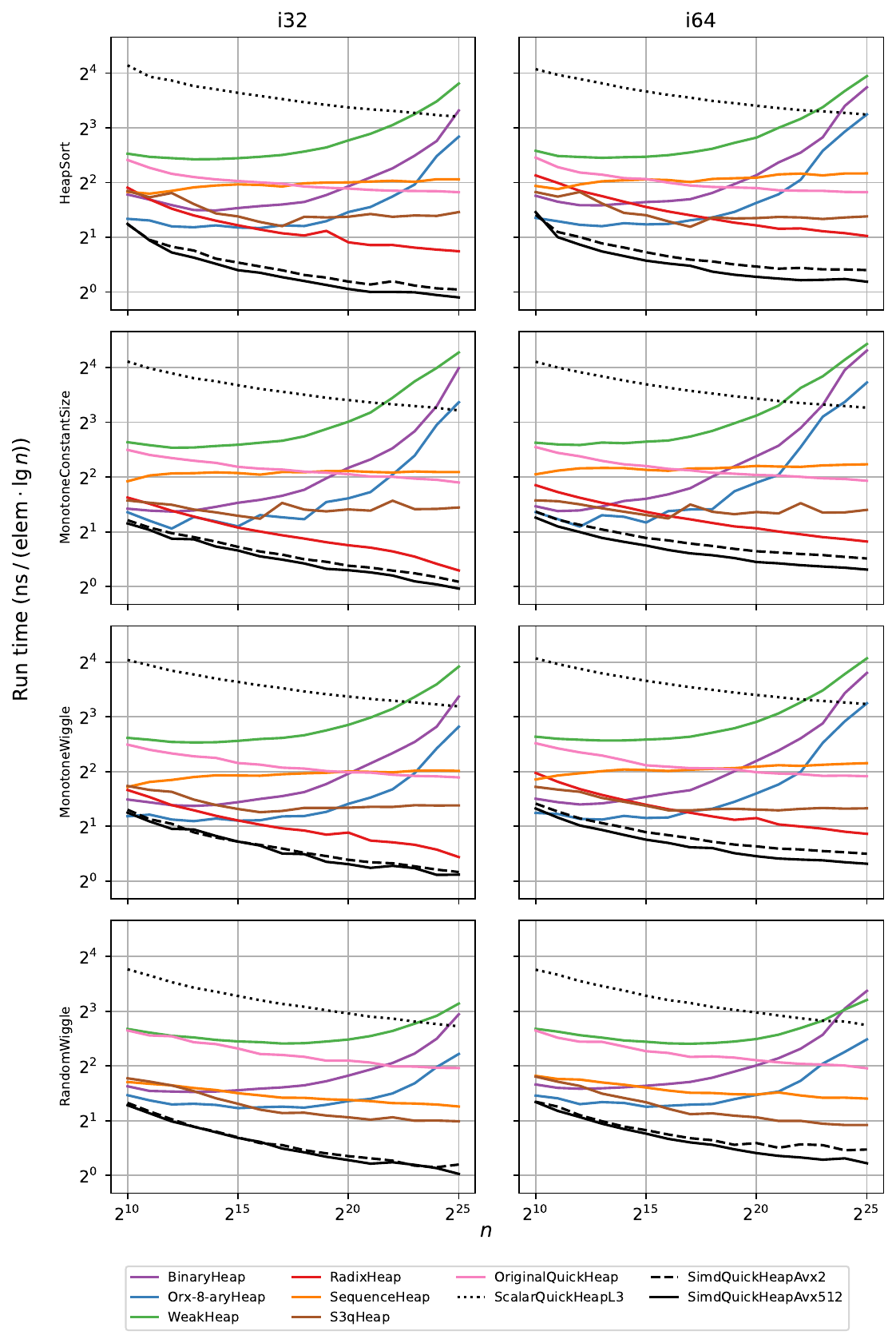}
	\caption{\label{fig:diffie-nanos}Log-log plots showing for a variety of workloads
		(rows, see \cref{sec:synth}), \texttt{i32} and \texttt{i64} inputs
		(columns), and increasing maximum number of elements stored in the data
		structure (x-axis), the time in nanoseconds to push and pop each element,
        divided by $\log_2 n$. That is, the time is relative to the $\Omega(\log_2 n)$ lower bound per
		operation.
	}
\end{figure}

\subsection{Synthetic benchmarks}\label{sec:synth}
We tested each implementation on several synthetic workloads, for both 32-bit and 64-bit keys.
We vary $n$, the maximum size of the heap, from $2^{10} = 1024$ to $2^{25} \approx 32$ million.

\parag{Workloads.} The \textsf{HeapSort} workload first generates $n$ random values and then measures how long it takes to push all of them and then pop all of them:
$\push^n\circ \pop^n$.
In the \textsf{RandomWiggle} workload, we grow the structure to size $n$ via a sequence of $3n$ operations $(\push\circ\pop\circ\push)^n$ and then empty it via $(\pop\circ\push\circ\pop)^n$.
The \textsf{MonotoneConstantSize} workload is initialized by growing the data structure to contain $n$ elements via $(\push\circ\pop\circ\push)^n$ as well.
Compared to just inserting $n$ values, this ensures that the internal state of each data structure is more randomized.
Then, we measure how long it takes to do $10 n$ pairs of $\pop\circ\push$ operations.
In \textsf{MonotoneWiggle} and \textsf{MonotoneConstantSize}, the pushed value is chosen as the last popped
value plus a uniform random constant between $0$ and $n$.
In the non-monotone \textsf{RandomWiggle}, the pushed value is a uniform random (32 or 64-bit) integer.
We note that this non-monotone case is degenerate, in that most values are pushed
to the front of the queue.

For each workload, the number $N$ of \Op{push} operations is the same as the number of \Op{pop} operations:
$N=n$ for \textsf{HeapSort}, $N=10n$ for
\textsf{MonotoneConstantSize}, and $N=3n$ for \textsf{MonotoneWiggle}
and \textsf{RandomWiggle}.
We report the average time \emph{per element}, which is the total running time divided by $N$.

\begin{table}[t]
	\caption{\label{tab:absolute}Average time in nanoseconds for pushing and popping each
      element for all implementations on the \texttt{MonotoneConstantSize} workload with \texttt{i64} keys, for increasing heap size $n$.
		The last column gives the number of L2 cache misses per element at the
        largest size $n=2^{25}$, excluding cache lines prefetched by the
        hardware prefetcher.
		The implementations above the rule are shown in \cref{fig:diffie-nanos}.
		The best value in each column is set in bold.}
	\centering
	\pgfplotstableread[col sep=comma]{table-sorted.dat}\datatable
	\pgfplotstablecreatecol[
		create col/set list={0,1,2,3,4,5,6,7,8,9,10,11,12,13,14,15,16,17,18,19,20,21}
	]{order}\datatable
	\pgfplotstabletypeset[
		sort, sort key={order}, sort cmp={float <},
		columns={label,nanos_1024,nanos_32768,nanos_1048576,nanos_33554432,%
				hw_cache_misses_33554432},
		fixed, fixed zerofill, precision=1, 1000 sep={\,},
		create on use/label/.style={create col/copy={displayname}},
		columns/label/.style={string type, column name={Implementation}, column type={l}},
		every column/.append style={column type={r}},
		highlight best/.style={
				postproc cell content/.append code={%
						\ifnum\pgfplotstablerow=#1\relax
							\pgfkeysalso{@cell content/.add={\bfseries\boldmath}{}}%
						\fi
					},
			},
		columns/nanos_1024/.style={column name={$n=2^{10}$}, highlight best=10},
		columns/nanos_32768/.style={column name={$2^{15}$}, highlight best=9},
		columns/nanos_1048576/.style={column name={$2^{20}$}, highlight best=9},
		columns/nanos_33554432/.style={column name={$2^{25}$}, highlight best=9},
		columns/hw_cache_misses_33554432/.style={column name={$n=2^{25}$}, precision=2, highlight best=21},
		every head row/.style={
				before row={\toprule
						& \multicolumn{4}{c}{Time (\si{\nano\second\per elem})}
						& \multicolumn{1}{c}{L2 misses/elem}\\
						\cmidrule(lr){2-5}\cmidrule(lr){6-6}
					},
				after row=\midrule,
			},
		every row no 9/.style={after row=\midrule},
		every last row/.style={after row=\bottomrule},
	]\datatable
\end{table}

\begin{figure}[tb]
	\centering
	\makebox[\linewidth]{
		\includegraphics[width=\linewidth]{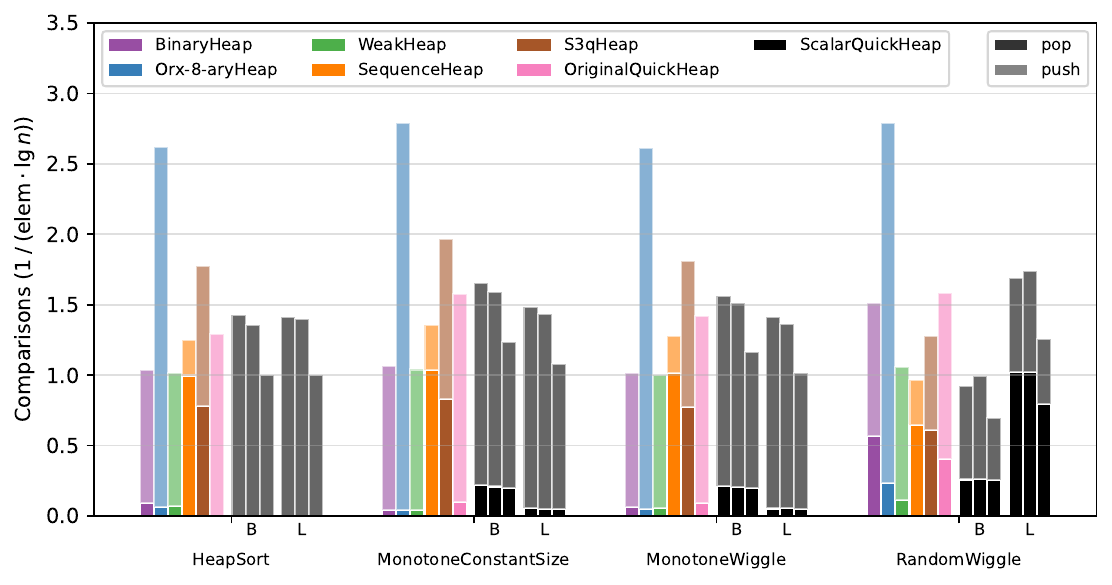}
	}
	\caption{\label{fig:comparisons}The number of comparisons per operation
		relative to the $\lg_2 n$ lower bound in different workloads. The number of comparisons during
		push operations is highlighted on the bottom, with the number of comparisons during pop
		operations shown on top. The SimdQuickHeap variants differ in using binary
		search (B) or linear search (L), and the pivot selection (in order: random, median of 3,
		true median oracle).}
\end{figure}

\parag{Results.}
\Cref{fig:diffie-nanos} shows the average time per element normalized by $\log_2 n$ for the best performing priority queue implementations on various workloads.
The normalization results in the time per \emph{fundamental operation}, since $\log_2$ is the minimum number of comparisons needed per element when the queue contains $n$ elements.

The absolute time per element for all implementations on the \texttt{MonotoneConstantSize} workload is given in \Cref{tab:absolute}.
We see that the \textsf{BinaryHeap} and \textsf{$d$-aryHeap} slow down
significantly as the data grows beyond the size of the CPU caches. The
\textsf{$d$-aryHeap} is consistently around $1.4\times$ as fast as the \textsf{BinaryHeap}.
For large n, the \textsf{WeakHeap} is surprisingly competitive with the \textsf{BinaryHeap}.

The \textsf{RadixHeap} benefits from high cache-locality and gets faster rather
than slower (relative to the lower bound) as $n$ grows. It is faster for 32-bit than for 64-bit data.
Similarly, as predicted by the theory, the \textsf{SequenceHeap} is I/O-efficient and does
not slow down on larger inputs.
The more modern and optimized \textsf{S3qHeap} is up to $2\times$
faster for 64-bit data.

Our reimplementation of the original QuickHeap is on-par with the sequence heap,
while the scalar version of our new heap is $3\times$ \emph{slower}, possibly
because the original QuickHeap uses more efficient in-place partitioning.
The \textsf{SimdQuickHeap} is around $8\times$ as fast as the scalar version.
Like the other engineered heaps, it gets faster relative to the lower bound as $n$ increases,
indicating that the constant overhead of each \Op{push} and \Op{pop} is relatively large compared to the $\log_2 n$ pivot steps required for each element.
The AVX-512 version is up to $1.2\times$ as fast as the AVX2 version, and up to $2\times$ as fast as the \textsf{S3qHeap}.
Surprisingly, it also significantly outperforms the \textsf{RadixHeap} in all tested cases.

The \textsf{SimdQuickHeap} reaches below $\log_2 n$ nanoseconds per $\pop\circ\push$ pair.
This means that it requires around $1$ nanosecond for each comparison in the lower bound.
Since each nanosecond corresponds to $3.7$ clock cycles, each of which can
contain multiple SIMD instructions
on $8$ words, the overhead of moving data around is still very large compared to just the comparisons.

\parag{Number of comparisons.}
On the \textsf{MonotoneConstantSize} workload (\cref{fig:comparisons}), the \textsf{BinaryHeap} needs less than $1.1 \log_2 n$ comparisons per element, while the \textsf{$8$-aryHeap} needs around $2.7 \log_2 n$ comparisons.
The \textsf{BinaryHeap} needs more comparisons on degenerate input, while the \textsf{WeakHeap} consistently needs $1.0\log_2 n$.
The \textsf{ScalarQuickHeap} needs $1.7 \log_2 n$ comparisons when using random pivots in combination with a linear scan over the pivots to insert elements.
This decreases to $1.2\log_2 n$ when using an oracle that gives the exact median for free, showing that there is some room for improvement by using a more accurate partitioning scheme.

\parag{Running time distribution.}
When running \textsf{SimdQuickHeap} with $n=2^{25}$ on the \textsf{MonotoneConstantSize} workload, $21\%$ of the time is spent pushing elements, with $14\%$ (of the total) scanning the list of pivots to find the right layer to append to.
The other $78\%$ of time is spent on pop: $59\%$ partitioning the input, $11\%$
finding the position of the smallest element in the bottom bucket, and $3\%$
removing and returning this element.

\parag{Memory usage.}
For most workloads, the total capacity of the vectors allocated by the \textsf{SimdQuickHeap} is around $5\times$ the space required just to hold the elements.
For the \textsf{MonotoneConstantSize} workload, this factor grows to around $11$.
The increased memory consumption is likely due to the higher number of operations ($10n$ vs.~$\leq3n$), so that the data structure has more time to accumulate large temporary buckets due to bad pivots.
In comparison, binary heaps and the original QuickHeap can be implemented on top of a growing vector, resulting in an overhead of just the geometric growing factor (typically 2).
Possible mitigations to lower the memory consumption would be to periodically shrink vectors with excessive capacity or to use a more memory-efficient data structure like a list of (reusable) blocks.

\begin{table}[t]
	\caption{\label{tab:graph-instances}Number of edges and vertices of the
		graph instances that have been used for benchmarking, as well as the
		median and maximal edge weight.
		$N_\text{D}$ and $N_\text{JP}$ are the numbers of processed queue elements during Dijkstra's and Jarn{\'i}k--Prim's algorithm, respectively.}
	\centering
	\begin{tabular}{llrrrrrrr}
		\toprule
		\multicolumn{1}{c}{Name} & \multicolumn{1}{c}{Description} & \multicolumn{1}{c}{$|V|$} & \multicolumn{1}{c}{$|E|$} & \multicolumn{1}{c}{max} & \multicolumn{1}{c}{median} & \multicolumn{1}{c}{max} & \multicolumn{1}{c}{$N_\text{D}$} & \multicolumn{1}{c}{$N_\text{JP}$} \\
		\multicolumn{1}{c}{} & \multicolumn{1}{c}{} & \multicolumn{1}{c}{($\times 10^6$)} & \multicolumn{1}{c}{($\times 10^6$)} & \multicolumn{1}{c}{degree} & \multicolumn{1}{c}{weight} & \multicolumn{1}{c}{weight} & \multicolumn{1}{c}{($\times 10^6$)} & \multicolumn{1}{c}{($\times 10^6$)} \\
		\toprule
		CAL                 & Rd. California  & 1               & 4               & 8      & 3115   & 54k    & 2               & 2               \\
		CTR                 & Rd. Central USA & 14              & 34              & 9      & 3467   & 54k    & 15              & 17              \\
		GER                 & Rd. Germany     & 20              & 41              & 9      & 39     & 62k    & 21              & 22              \\
		USA                 & Rd. USA         & 23              & 58              & 9      & 3473   & 92k    & 26              & 29              \\
		$\textrm{RHG}_{20}$ & Rand. hyperbolic & 1               & 10              & 90k    & 230    & 999    & 2               & 5               \\
		$\textrm{RHG}_{22}$ & Rand. hyperbolic & 4               & 41              & 941k   & 230    & 999    & 7               & 21              \\
		$\textrm{RHG}_{24}$ & Rand. hyperbolic & 16              & 159             & 595k   & 232    & 999    & 28              & 80              \\
		\bottomrule
	\end{tabular}
\end{table}

\begin{figure}[t]
	\centering
	\includegraphics[width=\linewidth,keepaspectratio]{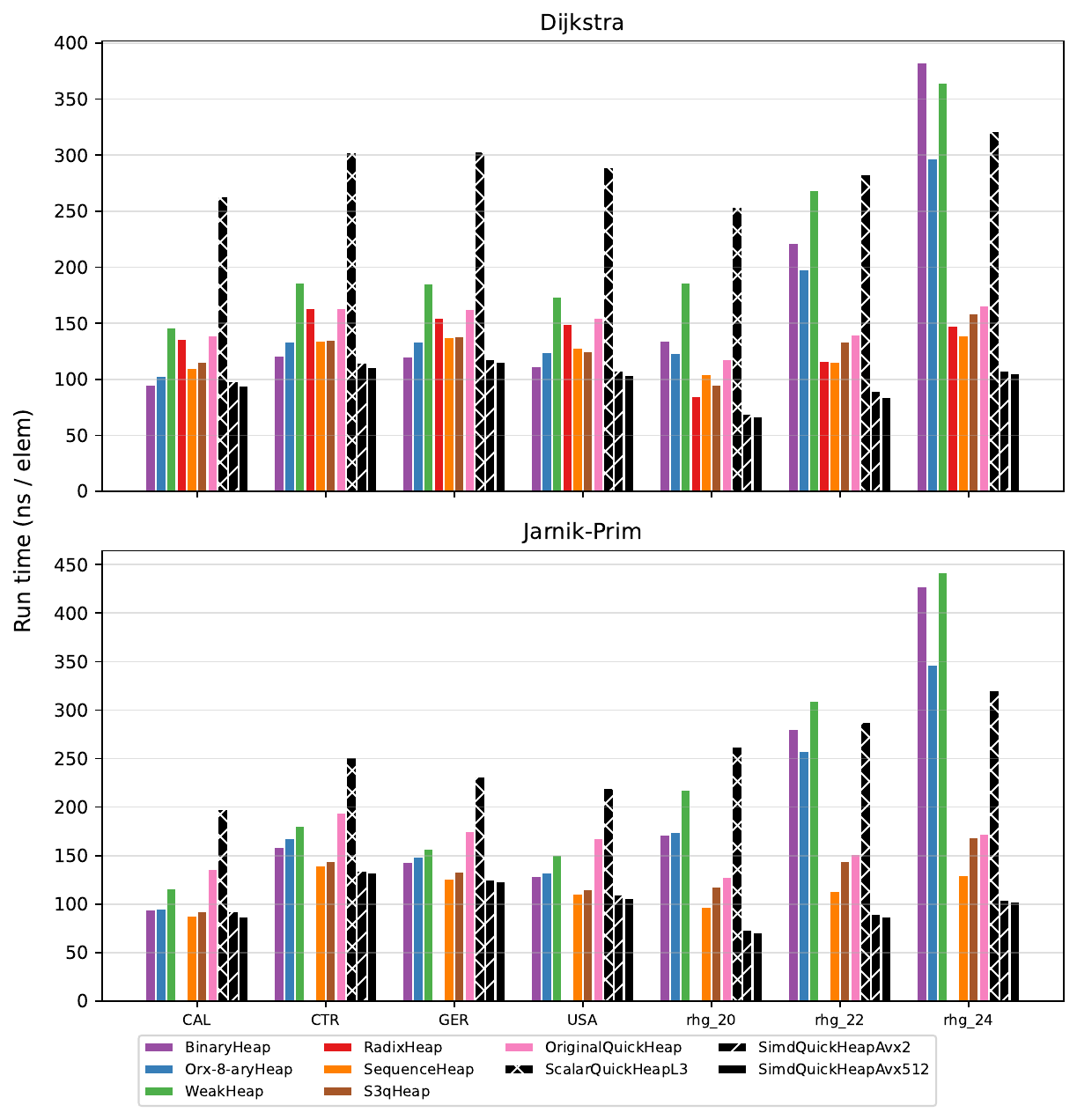}%
	\vspace{-1em}
	\caption{\label{fig:diffie-graphs}Time per processed queue element for the different heaps on multiple graph instances.
		The radix heap requires non-decreasing keys and is excluded from the Jarn{\'i}k--Prim benchmark.
	}
\end{figure}

\subsection{Graph benchmarks}
Two textbook graph algorithms that require a priority queue are Dijkstra's
algorithm~\cite{dijkstra59} to compute shortest paths on a weighted directed
graph with non-negative edge weights, and Jarn{\'i}k--Prim's
algorithm~\cite{jarnik30,prim57} to compute minimum-spanning-trees on weighted, undirected graphs.
Keys are non-decreasing for Dijkstra's algorithm, while keys may decrease for Jarn{\'i}k--Prim's algorithm, thus excluding the radix heap.
In this section, we compare running times of the two algorithms on multiple graph instances using different priority queues.

\parag{Setup.} We use 32-bit identifiers for vertices/edges and 32-bit non-negative edge weights.
We pack these into a single 64-bit unsigned integer that we use as elements for the priority queues.
For Dijkstra's algorithm, we re-insert a vertex when a new shorter path is found and discard stale versions of the vertex when popped.
We do not implement a variant that updates vertices instead of re-inserting them since re-inserting generally results in better performance \cite{dijkstra-decrease-key} and no priority queue implementation that supports \Op{decreaseKey} offers competitive performance in the synthetic benchmarks (see \Cref{tab:absolute}).

\parag{Graph types.}
We used different types of graphs to benchmark the performance of the
SimdQuickHeap (\cref{tab:graph-instances}).
We test on four road networks of varying sizes with edge weights representing travel times, all but the German\footnote{\url{https://i11www.iti.kit.edu/resources/roadgraphs.php}} network downloaded from 9th DIMACS implementation challenge\footnote{\url{http://www.diag.uniroma1.it/challenge9/download.shtml}}.
We also test on random hyperbolic graphs generated with KaGen~\cite{funkeCommunicationfreeMassivelyDistributed2019} with $2^{20}$, $2^{22}$, and $2^{24}$
nodes, an average degree of 16 and a power law exponent of $\gamma = 2.3$.
Edge weights represent the hyperbolic distance.

\parag{Results.}
We report the average running time per processed queue element, i.e., the number of elements that are pushed into and subsequently popped from the queue.
In \cref{fig:diffie-graphs}, we see that the SimdQuickHeap is the fastest on all graphs for both Dijkstra and Jarn{\'i}k--Prim workloads.
Speedups on road networks are small since the cache misses involved with traversing the graph are likely the bottleneck.
On the hyperbolic graphs, we achieve speedup factors of up to $4\times$ compared to the \textsf{binary heap}.

\section{Conclusion}\label{conclusion}
Due to its conceptually simple design, the SimdQuickHeap allows an efficient implementation using SIMD instructions that has both a good theoretical complexity and great performance in practice.
As predicted by the improved I/O-complexity, the SimdQuickHeap vastly outperforms the binary and $d$-ary heap, as well as the weak heap.
This new design gives a $4\times$ speedup over I/O-efficient structures such as the sequence heap (2000)~\cite{sequence-heap} and QuickHeap (2010)~\cite{quickheap}, and a $1.4\times$ to $2.8\times$ speedup over the much more recent superscalar sample queue~\cite{superscalar-queue}.
Maybe more surprisingly, the SimdQuickHeap is also up to $1.7\times$ as fast as the non-comparison-based radix heap.

\parag{Future work.}
Future work will be to implement the recently introduced rebalancing strategies
of~\cite{partition-based-simple-heaps} to limit the number of buckets,
and to evaluate if and how much they affect performance in both (currently) degenerate and non-degenerate cases.
This will also make the running time and I/O-complexity hold unconditionally.
Furthermore, the current $\BigO(\frac 1B \log_2 \frac nM)$ I/O-complexity is missing the more efficient $\log_{M/B}$ that is obtained by multi-way merging/splitting in the sequence heap and superscalar sample queue, and so the question is whether the QuickHeap naturally extends to multi-way splitting.

Another option is to make a SIMD-optimized version of the radix-heap, as it has a similar structure with exponentially growing buckets.

From the engineering side, optimizing the partitioning might provide further gains.
A drawback of the SimdQuickHeap is the larger memory usage compared to the original QuickHeap and simple binary heaps.
Using a list of reusable blocks rather than single vectors could reduce memory usage.
Furthermore, the implementation could be expanded to support key-value pairs that are both 64 bits,
bulk-insertions, and multi-threading.


\enlargethispage{2em}
\bibliographystyle{plainurl}
\bibliography{bibliography}

@inproceedings{array-heap,
  title = {An Experimental Study of Priority Queues in External Memory},
  booktitle = {Algorithm Engineering},
  author = {Brengel, Klaus and Crauser, Andreas and Ferragina, Paolo and Meyer, Ulrich},
  year = {1999},
  pages = {345--359},
  publisher = {Springer Berlin Heidelberg},
  issn = {1611-3349},
  doi = {10.1007/3-540-48318-7_27},
  isbn = {978-3-540-48318-2}
}

@article{b-heap,
  title = {You're Doing It Wrong},
  author = {Kamp, Poul-Henning},
  year = {2010},
  month = jul,
  journal = {Communications of the ACM},
  volume = {53},
  number = {7},
  pages = {55--59},
  publisher = {Association for Computing Machinery (ACM)},
  issn = {1557-7317},
  doi = {10.1145/1785414.1785434}
}

@inproceedings{back-to-basics-priority-queues,
  title = {A Back-to-Basics Empirical Study of Priority Queues},
  booktitle = {2014 Proceedings of the Sixteenth Workshop on Algorithm Engineering and Experiments ({{ALENEX}})},
  author = {Larkin, Daniel H. and Sen, Siddhartha and Tarjan, Robert E.},
  year = {2013},
  month = dec,
  pages = {61--72},
  publisher = {{Society for Industrial and Applied Mathematics}},
  doi = {10.1137/1.9781611973198.7},
  isbn = {978-1-61197-319-8}
}

@inproceedings{benchmarking-weak-heaps,
  title = {Policy-Based Benchmarking of Weak Heaps and Their Relatives,},
  booktitle = {Experimental Algorithms},
  author = {Bruun, Asger and Edelkamp, Stefan and Katajainen, Jyrki and Rasmussen, Jens},
  year = {2010},
  pages = {424--435},
  publisher = {Springer Berlin Heidelberg},
  issn = {1611-3349},
  doi = {10.1007/978-3-642-13193-6_36},
  isbn = {978-3-642-13193-6}
}

@article{binomial-heap,
  title = {A Data Structure for Manipulating Priority Queues},
  author = {Vuillemin, Jean},
  year = {1978},
  month = apr,
  journal = {Communications of the ACM},
  volume = {21},
  number = {4},
  pages = {309--315},
  publisher = {Association for Computing Machinery (ACM)},
  issn = {1557-7317},
  doi = {10.1145/359460.359478}
}

@article{binomial-heap-analysis,
  title = {Implementation and Analysis of Binomial Queue Algorithms},
  author = {Brown, Mark R.},
  year = {1978},
  month = aug,
  journal = {SIAM Journal on Computing},
  volume = {7},
  number = {3},
  pages = {298--319},
  publisher = {Society for Industrial \& Applied Mathematics (SIAM)},
  issn = {1095-7111},
  doi = {10.1137/0207026}
}

@article{blumTimeBoundsSelection1973,
  title = {Time Bounds for Selection},
  author = {Blum, Manuel and Floyd, Robert W. and Pratt, Vaughan and Rivest, Ronald L. and Tarjan, Robert E.},
  year = {1973},
  month = aug,
  journal = {Journal of Computer and System Sciences},
  volume = {7},
  number = {4},
  pages = {448--461},
  issn = {0022-0000},
  doi = {10.1016/S0022-0000(73)80033-9},
  urldate = {2026-04-23}
}

@misc{boost,
  title = {Boost {{C}}++ Libraries},
  author = {{Boost}},
  year = {2026},
  url = {https://www.boost.org/}
}

@article{bramasNovelHybridQuicksort2017,
  title = {A {{Novel Hybrid Quicksort Algorithm Vectorized}} Using {{AVX-512}} on {{Intel Skylake}}},
  author = {Bramas, Berenger},
  year = {2017},
  month = oct,
  journal = {International Journal of Advanced Computer Science and Applications (ijacsa)},
  volume = {8},
  number = {10},
  issn = {2156-5570},
  doi = {10.14569/IJACSA.2017.081044},
  urldate = {2026-04-24},
  langid = {english}
}

@inproceedings{brodal-queue,
  title = {Worst-Case Efficient Priority Queues},
  booktitle = {Proceedings of the Seventh Annual {{ACM-SIAM}} Symposium on {{Discrete}} Algorithms},
  author = {Brodal, Gerth St{\o}lting},
  year = {1996},
  month = jan,
  series = {{{SODA}} '96},
  pages = {52--58},
  publisher = {{Society for Industrial and Applied Mathematics}},
  url = {https://dl.acm.org/doi/10.5555/313852.313883},
  urldate = {2026-07-08},
  isbn = {978-0-89871-366-4}
}

@inproceedings{bucket-heap,
  title = {Cache-Oblivious Data Structures and Algorithms for Undirected Breadth-First Search and Shortest Paths},
  booktitle = {Algorithm Theory - {{SWAT}} 2004},
  author = {Brodal, Gerth St{\o}lting and Fagerberg, Rolf and Meyer, Ulrich and Zeh, Norbert},
  year = {2004},
  pages = {480--492},
  publisher = {Springer Berlin Heidelberg},
  issn = {1611-3349},
  doi = {10.1007/978-3-540-27810-8_41},
  isbn = {978-3-540-27810-8}
}

@article{bucket-queue,
  title = {Algorithm 360: Shortest-Path Forest with Topological Ordering [{{H}}]},
  author = {Dial, Robert B.},
  year = {1969},
  month = nov,
  journal = {Communications of the ACM},
  volume = {12},
  number = {11},
  pages = {632--633},
  publisher = {Association for Computing Machinery (ACM)},
  issn = {1557-7317},
  doi = {10.1145/363269.363610}
}

@article{buffer-heap,
  title = {Cache-Oblivious Buffer Heap and Cache-Efficient Computation of Shortest Paths in Graphs},
  author = {Chowdhury, Rezaul A. and Ramachandran, Vijaya},
  year = {2018},
  month = jan,
  journal = {ACM Transactions on Algorithms},
  volume = {14},
  number = {1},
  pages = {1--33},
  publisher = {Association for Computing Machinery (ACM)},
  issn = {1549-6333},
  doi = {10.1145/3147172}
}

@inproceedings{buffer-heap-spaa,
  title = {Cache-Oblivious Shortest Paths in Graphs Using Buffer Heap},
  booktitle = {Proceedings of the Sixteenth Annual {{ACM}} Symposium on {{Parallelism}} in Algorithms and Architectures},
  author = {Chowdhury, Rezaul Alam and Ramachandran, Vijaya},
  year = {2004},
  month = jun,
  series = {{{SPAA04}}},
  pages = {245--254},
  publisher = {ACM},
  doi = {10.1145/1007912.1007949}
}

@article{cache-oblivious-algorithms,
  title = {Cache-Oblivious Algorithms},
  author = {Frigo, Matteo and Leiserson, Charles E. and Prokop, Harald and Ramachandran, Sridhar},
  year = {2012},
  month = jan,
  journal = {ACM Transactions on Algorithms},
  volume = {8},
  number = {1},
  pages = {1--22},
  publisher = {Association for Computing Machinery (ACM)},
  issn = {1549-6333},
  doi = {10.1145/2071379.2071383}
}

@inproceedings{cache-oblivious-priority-queue,
  title = {Cache-Oblivious Priority Queue and Graph Algorithm Applications},
  booktitle = {Proceedings of the Thiry-Fourth Annual {{ACM}} Symposium on {{Theory}} of Computing},
  author = {Arge, Lars and Bender, Michael A. and Demaine, Erik D. and {Holland-Minkley}, Bryan and Munro, J. Ian},
  year = {2002},
  month = may,
  series = {{{STOC02}}},
  pages = {268--276},
  publisher = {ACM},
  doi = {10.1145/509907.509950}
}

@article{cache-oblivious-priority-queue-extended,
  title = {An Optimal Cache-oblivious Priority Queue and Its Application to Graph Algorithms},
  author = {Arge, Lars and Bender, Michael A. and Demaine, Erik D. and Holland-Minkley, Bryan and Ian Munro, J.},
  year = {2007},
  month = jan,
  journal = {SIAM Journal on Computing},
  volume = {36},
  number = {6},
  pages = {1672--1695},
  publisher = {Society for Industrial \& Applied Mathematics (SIAM)},
  issn = {1095-7111},
  doi = {10.1137/s0097539703428324}
}

@inproceedings{cascade-heap,
  title = {Cascade Heap: {{Towards}} Time-Optimal Extractions},
  booktitle = {Computer Science -- Theory and Applications},
  author = {Babenko, Maxim and Kolesnichenko, Ignat and Smirnov, Ivan},
  year = {2017},
  pages = {62--70},
  publisher = {Springer International Publishing},
  issn = {1611-3349},
  doi = {10.1007/978-3-319-58747-9_8},
  isbn = {978-3-319-58747-9}
}

@inproceedings{complexity-of-priority-queue-maintenance,
  title = {The Complexity of Priority Queue Maintenance},
  booktitle = {Proceedings of the Ninth Annual {{ACM}} Symposium on {{Theory}} of Computing - {{STOC}} '77},
  author = {Brown, Mark R.},
  year = {1977},
  series = {Stoc '77},
  pages = {42--48},
  publisher = {ACM Press},
  doi = {10.1145/800105.803394}
}

@article{dary-heap,
  title = {Priority Queues with Update and Finding Minimum Spanning Trees},
  author = {Johnson, Donald B.},
  year = {1975},
  month = dec,
  journal = {Information Processing Letters},
  volume = {4},
  number = {3},
  issn = {0020-0190},
  doi = {10.1016/0020-0190(75)90001-0},
  urldate = {2023-01-04},
  langid = {english}
}

@inproceedings{decrease-key-lower-bound,
  title = {Why Some Heaps Support Constant-Amortized-Time Decrease-Key Operations, and Others Do Not},
  booktitle = {Automata, Languages, and Programming},
  author = {Iacono, John and {\"O}zkan, {\"O}zg{\"u}r},
  year = {2014},
  pages = {637--649},
  publisher = {Springer Berlin Heidelberg},
  issn = {1611-3349},
  doi = {10.1007/978-3-662-43948-7_53},
  isbn = {978-3-662-43948-7}
}

@inproceedings{decreasekey-expensive,
  title = {{{DecreaseKeys}} Are Expensive for External Memory Priority Queues},
  booktitle = {Proceedings of the 49th Annual {{ACM SIGACT}} Symposium on Theory of Computing},
  author = {Eenberg, Kasper and Larsen, Kasper Green and Yu, Huacheng},
  year = {2017},
  month = jun,
  series = {Stoc '17},
  pages = {1081--1093},
  publisher = {ACM},
  doi = {10.1145/3055399.3055437}
}

@misc{dijkstra-decrease-key,
  title = {Priority {{Queues}} and {{Dijkstra}}'s {{Algorithm}}},
  author = {Chen, Mo and Chowdhury, Rezaul Alam and Ramachandran, Vijaya and Roche, David Lan and Tong, Lingling},
  url = {https://www.cs.utexas.edu/ftp/techreports/tr07-54.pdf},
  year = 2007,
  langid = {english}
}

@article{dijkstra59,
  title = {A Note on Two Problems in Connexion with Graphs},
  author = {Dijkstra, Edsger W},
  year = {1959},
  month = dec,
  journal = {Numerische Mathematik},
  volume = {1},
  number = {1},
  pages = {269--271},
  publisher = {{Springer Science and Business Media LLC}},
  issn = {0945-3245},
  doi = {10.1007/bf01386390}
}

@article{engineering-weak-heaps,
  title = {Weak Heaps Engineered},
  author = {Edelkamp, Stefan and Elmasry, Amr and Katajainen, Jyrki},
  year = {2013},
  month = nov,
  journal = {Journal of Discrete Algorithms},
  volume = {23},
  pages = {83--97},
  publisher = {Elsevier BV},
  issn = {1570-8667},
  doi = {10.1016/j.jda.2013.07.002}
}

@article{equivalence-priority-queues-sorting,
  title = {Equivalence between Priority Queues and Sorting},
  author = {Thorup, Mikkel},
  year = {2007},
  month = dec,
  journal = {Journal of the ACM},
  volume = {54},
  number = {6},
  pages = {28},
  publisher = {Association for Computing Machinery (ACM)},
  issn = {1557-735X},
  doi = {10.1145/1314690.1314692}
}

@article{external-memory-priority-queue-decreasekey,
  title = {A Faster External Memory Priority Queue with {{DecreaseKeys}}},
  author = {Jiang, Shunhua and Larsen, Kasper Green},
  year = {2018},
  journal = {arXiv},
  doi = {10.48550/ARXIV.1806.07598},
  copyright = {arXiv.org perpetual, non-exclusive license}
}

@inproceedings{external-memory-priority-queue-optimal-insertions,
  title = {External-{{Memory Priority Queues}} with {{Optimal Insertions}}},
  booktitle = {33rd {{Annual European Symposium}} on {{Algorithms}} ({{ESA}} 2025)},
  author = {Brodal, Gerth St{\o}lting and Goodrich, Michael T. and Iacono, John and Lo, Jared and Meyer, Ulrich and Pagan, Victor and Sitchinava, Nodari and Svenning, Rolf},
  year = {2025},
  series = {Leibniz {{International Proceedings}} in {{Informatics}} ({{LIPIcs}})},
  volume = {351},
  pages = {5:1--5:14},
  publisher = {Schloss Dagstuhl -- Leibniz-Zentrum f{\"u}r Informatik},
  issn = {1868-8969},
  doi = {10.4230/LIPIcs.ESA.2025.5},
  urldate = {2026-07-08},
  isbn = {978-3-95977-395-9}
}

@article{fibonacci-heap,
  title = {Fibonacci Heaps and Their Uses in Improved Network Optimization Algorithms},
  author = {Fredman, Michael L. and Tarjan, Robert Endre},
  year = {1987},
  month = jul,
  journal = {Journal of the ACM},
  volume = {34},
  number = {3},
  pages = {596--615},
  publisher = {Association for Computing Machinery (ACM)},
  issn = {1557-735X},
  doi = {10.1145/28869.28874}
}

@misc{fibonacci-heap-crate,
  title = {Fibonacci Heap in Rust},
  author = {{xvi-xv-xii-ix-xxii-ix-xiv}},
  year = {2026},
  url = {https://github.com/xvi-xv-xii-ix-xxii-ix-xiv/fibonacci_heap}
}

@article{funkeCommunicationfreeMassivelyDistributed2019,
  title = {Communication-Free Massively Distributed Graph Generation},
  author = {Funke, Daniel and Lamm, Sebastian and Meyer, Ulrich and Penschuck, Manuel and Sanders, Peter and Schulz, Christian and Strash, Darren and {von Looz}, Moritz},
  year = {2019},
  month = sep,
  journal = {Journal of Parallel and Distributed Computing},
  volume = {131},
  pages = {200--217},
  issn = {0743-7315},
  doi = {10.1016/j.jpdc.2019.03.011},
  urldate = {2024-03-23}
}

@inproceedings{funnel-heap,
  title = {Funnel {{Heap-A Cache Oblivious Priority Queue}}},
  booktitle = {Algorithms and {{Computation}}},
  author = {Brodal, Gerth St{\o}lting and Fagerberg, Rolf},
  year = {2002},
  pages = {219--228},
  publisher = {Springer},
  doi = {10.1007/3-540-36136-7_20},
  isbn = {978-3-540-36136-7},
  langid = {english}
}

@inproceedings{fusion-trees,
  title = {{{BLASTING}} through the Information Theoretic Barrier with {{FUSION TREES}}},
  booktitle = {Proceedings of the Twenty-Second Annual {{ACM}} Symposium on {{Theory}} of Computing - {{STOC}} '90},
  author = {Fredman, M. L. and Willard, D. E.},
  year = {1990},
  series = {Stoc '90},
  pages = {1--7},
  publisher = {ACM Press},
  doi = {10.1145/100216.100217}
}

@article{heapsort,
  title = {Algorithm 232: {{Heapsort}}},
  author = {Williams, J.W.J.},
  year = {1964},
  month = jun,
  journal = {Communications of the ACM},
  volume = {7},
  number = {6},
  pages = {347--348},
  publisher = {Association for Computing Machinery (ACM)},
  issn = {1557-7317},
  doi = {10.1145/512274.3734138}
}

@article{improved-weak-heap,
  title = {The Weak-Heap Data Structure: {{Variants}} and Applications},
  author = {Edelkamp, Stefan and Elmasry, Amr and Katajainen, Jyrki},
  year = {2012},
  month = oct,
  journal = {Journal of Discrete Algorithms},
  volume = {16},
  pages = {187--205},
  publisher = {Elsevier BV},
  issn = {1570-8667},
  doi = {10.1016/j.jda.2012.04.010}
}

@article{io-complexity-sorting,
  title = {The Input/Output Complexity of Sorting and Related Problems},
  author = {Aggarwal, Alok and Vitter, Jeffrey, S.},
  year = {1988},
  month = sep,
  journal = {Communications of the ACM},
  volume = {31},
  number = {9},
  pages = {1116--1127},
  publisher = {Association for Computing Machinery (ACM)},
  issn = {1557-7317},
  doi = {10.1145/48529.48535}
}

@article{boruvka26,
  title = {{O jist{\'e}m probl{\'e}mu minim{\'a}ln{\'i}m (On a certain problem of minimization)}},
  author = {Bor{\r u}vka, Otakar},
  year = {1926},
  pages = {36--58},
  url = {https://dml.cz/handle/10338.dmlcz/500114},
  journal = {Práce Moravské přírodovědecké společnosti, sv. III, spis 3, 1926},
  urldate = {2026-07-08},
  langid = {czech}
}

@article{jarnik30,
  title = {{O jist{\'e}m probl{\'e}mu minim{\'a}ln{\'i}m (On a certain problem of minimization)}},
  author = {Jarn{\'i}k, Vojt{\v e}ch},
  year = {1930},
  pages = {57--63},
  url = {https://dml.cz//handle/10338.dmlcz/500726},
  journal = {Práce Moravské přírodovědecké společnosti, sv. VI., spis 4, 1930},
  urldate = {2026-07-08},
  langid = {czech}
}

@misc{log-funnel-heap,
  title = {The Logarithmic Funnel Heap: {{An}} Efficient Priority Queue for Extremely Large Sets},
  author = {Loeffeld, Christian},
  year = {2023},
  publisher = {arXiv},
  doi = {10.48550/ARXIV.1705.10648},
  copyright = {Creative Commons Attribution Non Commercial No Derivatives 4.0 International}
}

@article{monotone-priority-queue-review,
  title = {Exploring Monotone Priority Queues for {{Dijkstra}} Optimization},
  author = {Costa, Jonas and Castro, Lucas and {de Freitas}, Rosiane},
  year = {2025},
  month = sep,
  journal = {RAIRO - Operations Research},
  volume = {59},
  number = {5},
  pages = {2419--2436},
  publisher = {EDP Sciences},
  issn = {2804-7303},
  doi = {10.1051/ro/2025082}
}

@incollection{optimal-incremental-sorting,
  title = {Optimal {{Incremental Sorting}}},
  booktitle = {2006 {{Proceedings}} of the {{Workshop}} on {{Algorithm Engineering}} and {{Experiments}} ({{ALENEX}})},
  author = {Paredes, Rodrigo and Navarro, Gonzalo},
  year = {2006},
  month = jan,
  series = {Proceedings},
  pages = {171--182},
  publisher = {{Society for Industrial and Applied Mathematics}},
  doi = {10.1137/1.9781611972863.16},
  urldate = {2026-07-08}
}

@misc{orx-priority-queue,
  title = {Orx-Priority-Queue: {{Priority}} Queue Traits and Efficient d-Ary Heap Implementations},
  author = {Arikan, Ugur},
  year = {2023},
  url = {https://github.com/orxfun/orx-priority-queue/}
}

@article{pairing-heap,
  title = {The Pairing Heap: {{A}} New Form of Self-Adjusting Heap},
  author = {Fredman, Michael L. and Sedgewick, Robert and Sleator, Daniel D. and Tarjan, Robert E.},
  year = {1986},
  month = nov,
  journal = {Algorithmica. An International Journal in Computer Science},
  volume = {1},
  number = {1--4},
  pages = {111--129},
  publisher = {{Springer Science and Business Media LLC}},
  issn = {1432-0541},
  doi = {10.1007/bf01840439}
}

@article{partition-based-simple-heaps,
  title = {Partition-Based Simple Heaps},
  author = {Brodal, Gerth St{\o}lting and Iacono, John and Rysgaard, Casper Moldrup and Wild, Sebastian},
  year = {2026},
  journal = {arXiv},
  doi = {10.48550/ARXIV.2603.01206},
  copyright = {arXiv.org perpetual, non-exclusive license}
}

@misc{pheap-crate,
  title = {Pairing Heap},
  author = {{1crcbl}},
  year = {2021},
  url = {https://github.com/1crcbl/pheap-rs}
}

@article{prim57,
  title = {Shortest Connection Networks and Some Generalizations},
  author = {Prim, R. C.},
  year = {1957},
  journal = {The Bell System Technical Journal},
  volume = {36},
  number = {6},
  pages = {1389--1401},
  doi = {10.1002/j.1538-7305.1957.tb01515.x}
}

@incollection{priority-queue-survey,
  title = {A {{Survey}} on {{Priority Queues}}},
  booktitle = {Space-{{Efficient Data Structures}}, {{Streams}}, and {{Algorithms}}: {{Papers}} in {{Honor}} of {{J}}. {{Ian Munro}} on the {{Occasion}} of {{His}} 66th {{Birthday}}},
  author = {Brodal, Gerth St{\o}lting},
  year = {2013},
  pages = {150--163},
  publisher = {Springer},
  doi = {10.1007/978-3-642-40273-9_11},
  urldate = {2026-07-08},
  isbn = {978-3-642-40273-9},
  langid = {english}
}

@article{quickheap,
  title = {On Sorting, Heaps, and Minimum Spanning Trees},
  author = {Navarro, Gonzalo and Paredes, Rodrigo},
  year = {2010},
  month = mar,
  journal = {Algorithmica. An International Journal in Computer Science},
  volume = {57},
  number = {4},
  pages = {585--620},
  publisher = {{Springer Science and Business Media LLC}},
  issn = {1432-0541},
  doi = {10.1007/s00453-010-9400-6}
}

@article{quicksort,
  title = {Algorithm 64: {{Quicksort}}},
  author = {Hoare, C. A. R.},
  year = {1961},
  month = jul,
  journal = {Communications of the ACM},
  volume = {4},
  number = {7},
  pages = {321},
  publisher = {Association for Computing Machinery (ACM)},
  issn = {1557-7317},
  doi = {10.1145/366622.366644}
}

@article{radix-heap,
  title = {Faster Algorithms for the Shortest Path Problem},
  author = {Ahuja, Ravindra K. and Mehlhorn, Kurt and Orlin, James and Tarjan, Robert E.},
  year = {1990},
  month = apr,
  journal = {Journal of the ACM},
  volume = {37},
  number = {2},
  pages = {213--223},
  publisher = {Association for Computing Machinery (ACM)},
  issn = {1557-735X},
  doi = {10.1145/77600.77615}
}

@misc{radix-heap-cpp,
  title = {Radix-Heap},
  author = {Akiba, Takuya},
  year = {2015},
  url = {https://github.com/iwiwi/radix-heap}
}

@misc{radix-heap-crate,
  title = {Radix-Heap: {{Fast}} Monotone Priority Queues},
  author = {Pedersen, Mike},
  year = {2016},
  url = {https://github.com/mpdn/radix-heap}
}

@article{self-adjusting-heaps,
  title = {Efficiency of Self-Adjusting Heaps},
  author = {Sinnamon, Corwin and Tarjan, Robert E.},
  year = {2025},
  month = sep,
  journal = {ACM Transactions on Algorithms},
  volume = {21},
  number = {4},
  pages = {1--39},
  publisher = {Association for Computing Machinery (ACM)},
  issn = {1549-6333},
  doi = {10.1145/3708989}
}

@article{sequence-heap,
  title = {Fast Priority Queues for Cached Memory},
  author = {Sanders, Peter},
  year = {2000},
  month = dec,
  journal = {ACM Journal of Experimental Algorithmics},
  volume = {5},
  pages = {7},
  publisher = {Association for Computing Machinery (ACM)},
  issn = {1084-6654},
  doi = {10.1145/351827.384249}
}

@misc{sequence-heap-code,
  title = {Sequence Heap},
  author = {Sanders, Peter},
  year = {2000},
  url = {https://github.com/raphinesse/SequenceHeap}
}

@inproceedings{slim-smooth-heap-analysis,
  title = {A Tight Analysis of Slim Heaps and Smooth Heaps},
  booktitle = {Proceedings of the 2023 Annual {{ACM-SIAM}} Symposium on Discrete Algorithms ({{SODA}})},
  author = {Sinnamon, Corwin and Tarjan, Robert E.},
  year = {2023},
  month = jan,
  pages = {549--567},
  publisher = {{Society for Industrial and Applied Mathematics}},
  doi = {10.1137/1.9781611977554.ch24},
  isbn = {978-1-61197-755-4}
}

@article{stronger-quickheaps,
  title = {Stronger Quickheaps},
  author = {Navarro, Gonzalo and Paredes, Rodrigo and Poblete, Patricio V. and Sanders, Peter},
  year = {2011},
  month = jun,
  journal = {International Journal of Foundations of Computer Science},
  volume = {22},
  number = {04},
  pages = {945--969},
  publisher = {World Scientific Pub Co Pte Lt},
  issn = {1793-6373},
  doi = {10.1142/s0129054111008507}
}

@inproceedings{super-scalar-sample-sort,
  title = {Super Scalar Sample Sort},
  booktitle = {Algorithms -- {{ESA}} 2004},
  author = {Sanders, Peter and Winkel, Sebastian},
  year = {2004},
  pages = {784--796},
  publisher = {Springer Berlin Heidelberg},
  issn = {1611-3349},
  doi = {10.1007/978-3-540-30140-0_69},
  isbn = {978-3-540-30140-0}
}

@mastersthesis{superscalar-queue,
  title = {Superscalar Sample Queue: {{Engineering}} a Distribution-Based Priority Queue},
  author = {{von der Gr{\"u}n}, Raphael},
  year = {2021},
  url = {https://ae.iti.kit.edu/english/4296.php},
  school = {Karlsruhe Institute of Technology}
}

@misc{superscalar-queue-code,
  title = {Superscalar Sample Queue},
  author = {{von der Gr{\"u}n}, Raphael},
  year = {2021},
  url = {https://github.com/raphinesse/s3q}
}

@article{two-level-bucket-queue,
  title = {Shortest-Route Methods: 1. {{Reaching}}, Pruning, and Buckets},
  author = {Denardo, Eric V. and Fox, Bennett L.},
  year = {1979},
  month = feb,
  journal = {Operations Research},
  volume = {27},
  number = {1},
  pages = {161--186},
  publisher = {{Institute for Operations Research and the Management Sciences (INFORMS)}},
  issn = {1526-5463},
  doi = {10.1287/opre.27.1.161}
}

@inproceedings{van-emde-boas-trees,
  title = {Preserving Order in a Forest in Less than Logarithmic Time},
  booktitle = {16th Annual Symposium on Foundations of Computer Science (Sfcs 1975)},
  author = {{van Emde Boas}, P.},
  year = {1975},
  month = oct,
  publisher = {IEEE},
  doi = {10.1109/sfcs.1975.26}
}

@article{wassenbergVectorizedPerformanceportableQuicksort2022a,
  title = {Vectorized and Performance-Portable Quicksort},
  author = {Wassenberg, Jan and Blacher, Mark and Giesen, Joachim and Sanders, Peter},
  year = {2022},
  journal = {Software: Practice and Experience},
  volume = {52},
  number = {12},
  pages = {2684--2699},
  issn = {1097-024X},
  doi = {10.1002/spe.3142},
  urldate = {2026-04-24},
  copyright = {{\copyright} 2022 John Wiley \& Sons Ltd.},
  langid = {english}
}

@misc{weak-heap-crate,
  title = {Weak-Heap},
  author = {Starovoitov, Egor},
  year = {2022},
  url = {https://github.com/starovoid/weakheap}
}

@article{weak-heap-sort,
  title = {Weak-Heap Sort},
  author = {Dutton, Ronald D.},
  year = {1993},
  month = sep,
  journal = {Bit Numerical Mathematics},
  volume = {33},
  number = {3},
  pages = {372--381},
  publisher = {{Springer Science and Business Media LLC}},
  issn = {1572-9125},
  doi = {10.1007/bf01990520}
}

@article{xy-fast-trie,
  title = {Log-Logarithmic Worst-Case Range Queries Are Possible in Space {{$\Theta$}}({{N}})},
  author = {Willard, Dan E.},
  year = {1983},
  month = aug,
  journal = {Information Processing Letters},
  volume = {17},
  number = {2},
  pages = {81--84},
  publisher = {Elsevier BV},
  issn = {0020-0190},
  doi = {10.1016/0020-0190(83)90075-3}
}



\end{document}